\documentclass[11pt]{article}

\usepackage[active]{srcltx}
\usepackage{hyperref}
\usepackage{amsmath,amssymb}
\usepackage[all]{xy}
\usepackage{graphicx}
\usepackage{color}
 1 
\newtheorem{theorem}{Theorem}
\newtheorem{proposition}{Proposition}
\newtheorem{corollary}{Corollary}

\newtheorem{definition}{Definition}

\def\qed{\ifvmode\Realemovelastskip\fi
{\unskip\nobreak\hfil\penalty50\hbox{}\nobreak\hfil \hbox{\vrule
height1.2ex width1.2ex}\parfillskip=0pt \finalhyphendemerits=0
\par\smallskip}}
\def\qedr{\ifvmode\Realemovelastskip\fi
{\unskip\nobreak\hfil\penalty50\hbox{}\nobreak\hfil \hbox{
$\diamond$}\parfillskip=0pt \finalhyphendemerits=0
\par\smallskip}}
\def\ds{\displaystyle}
\newenvironment{proof}{\noindent{\sl Proof:~~~}}{\quad \qed}

\def\beq{\begin{equation}}
\def\eeq{\end{equation}}
\def\bea{\begin{eqnarray}}
\def\eea{\end{eqnarray}}
\def\beann{\begin{eqnarray*}}
\def\eeann{\end{eqnarray*}}
\def\beasn{\begin{sneqnarray}}
\def\eeasn{\end{sneqnarray}}
\def\ben{\begin{enumerate}}
\def\een{\end{enumerate}}
\def\bit{\begin{itemize}}
\def\eit{\end{itemize}}

\def\derpar#1#2{\displaystyle\frac{\partial{#1}}{\partial{#2}}}
\def\derpars#1#2#3{\displaystyle\frac{\partial^2{#1}}{\partial{#2}\partial{#3}}}

\def\W{{\cal W}}
\def\C{{\cal C}}

\def\vf{{\mathfrak{X}}}
\def\df{{\mit\Omega}}
\def\Lag{{\cal L}}
\def\Leg{{\cal FL}}

\def\d{{\rm d}}

\def\p{{\rm p}}

\def\Zahl{\mathbb{Z}}
\def\Real{\mathbb{R}}
\def\R{\mathbb{R}}

\def\pr{\operatorname{pr}}
\def\Tan{{\rm T}}

\def\inn{\mathop{i}\nolimits}
\def\Cinfty{{\rm C}^\infty}

\def\tabaddress#1{{\small\it\begin{tabular}[t]{c}#1
\\[1.2ex]\end{tabular}}}
\def\qed{\ifvmode\removelastskip\fi
{\unskip\nobreak\hfil\penalty50\hbox{}\nobreak\hfil \hbox{\vrule
height1.2ex width1.2ex}\parfillskip=0pt \finalhyphendemerits=0
\par\smallskip}}

\title{UNIFIED FORMALISM FOR THE GENERALIZED $k$th-ORDER HAMILTON-JACOBI PROBLEM}

\author{ 
{\sc  Leonardo Colombo\thanks{{\bf e}-{\it mail}: leo.colombo@icmat.es} }\\
{\sc  Manuel de Le\'on\thanks{{\bf e}-{\it mail}: mdeleon@icmat.es} }\\
\vspace{5mm}
   \tabaddress{Instituto de Ciencias Matem\'{a}ticas (CSIC-UAM-UC3M-UCM). \\
   C/ Nicol\'{a}s Cabrera 15. Campus Cantoblanco UAM. 28049 Madrid. Spain} \\
{\sc  Pedro Daniel Prieto-Mart\'\i nez\thanks{{\bf e}-{\it mail}: peredaniel@ma4.upc.edu} }\\
{\sc Narciso Rom\'an-Roy\thanks{{\bf e}-{\it mail}: nrr@ma4.upc.edu}}  \\
   \tabaddress{Departamento de Matem\'atica Aplicada IV.\\
Universitat Polit\`ecnica de Catalunya-Barcelona Tech.\\
   Edificio C-3, Campus Norte UPC.
   C/ Jordi Girona 1. 08034 Barcelona. Spain}
}

\begin{document}

\maketitle

\thispagestyle{empty}

\begin{abstract}
The geometric formulation of the Hamilton-Jacobi theory enables us
to generalize it to systems of higher-order ordinary differential
equations. In this work we introduce the unified
Lagrangian-Hamiltonian formalism for the geometric Hamilton-Jacobi
theory on higher-order autonomous dynamical systems described by
regular Lagrangian functions.
\end{abstract}

\medskip
\noindent
{\sl Key words}:
Hamilton-Jacobi equation, Higher-order systems, Skinner-Rusk formalism.

\smallskip
\noindent
\vbox{\raggedleft AMS s.\,c.\,(2010): 53C80, 70H20, 70H50.}\null

\markright{\rm  L. Colombo, M. de Le\'on, P.D. Prieto-Mart\'{\i}nez, N. Rom\'{a}n-Roy:
 \sl Generalized $k$th-order HJ}

\clearpage

\tableofcontents


\section{Introduction}

The geometric formulation of the Hamilton-Jacobi theory given in
\cite{GHJT} and \cite{dLMdDVa} enables us to generalize it to
systems of higher-order ordinary differential equations.
This generalization has been done recently for the Lagrangian and Hamiltonian formalism
of higher-order autonomous mechanical systems described by
regular Lagrangian functions \cite{CoMdLPrRo}.
The aim of this work is to give a unified Lagrangian-Hamiltonian version
of this theory for these kinds of systems, using the unified framework introduced
by Skinner and Rusk \cite{SKRK}.
The advantage of this formulation is that it compresses the Lagrangian and Hamiltonian
Hamilton--Jacobi problems into a single formalism which allows to recover both of them
in a simple way, and it is specially interesting when dealing with singular systems.

All the manifolds are real, second countable and $\Cinfty$. 
The maps and the structures are assumed to be $\Cinfty$. Sum over repeated indices is understood.

\section{Higher-order tangent bundles}

Let $Q$ be a $n$-dimensional manifold, and $k \in \Zahl^+$. The {\sl
$k$th-order tangent bundle of $Q$} is the $(k+1)n$-dimensional
manifold $\Tan^{k}Q$ made of the $k$-jets of the bundle $\pi \colon
\R \times Q \to \R$ with fixed source point $t = 0 \in \R$; that is,
$\Tan^kQ = J_0^k\pi$.

We have the following natural projections (for $r \leqslant k$):
$$
\begin{array}{rcl}
\rho^k_r \colon \Tan^kQ & \longrightarrow & \Tan^rQ \\
j^k_0\phi & \longmapsto & j^r_0\phi
\end{array} \quad ; \quad
\begin{array}{rcl}
\beta^k \colon \Tan^kQ & \longrightarrow & Q \\
j^k_0\phi & \longmapsto & \phi(0)
\end{array}
$$
where $j^k_0\phi$ denotes a point in $\Tan^kQ$; that is, the
equivalence class of a curve $\phi \colon I \subset \R \to Q$ by the
$k$-jet equivalence relation. Notice that $\rho^k_0 = \beta^k$,
where $\Tan^0Q$ is canonically identified with $Q$, and $\rho^k_k =
{\rm Id}_{\Tan^kQ}$. Observe also that $\rho^l_s \circ \rho^r_l =
\rho^r_s$, for $0 \leqslant s \leqslant l \leqslant r \leqslant k$.

If $\phi \colon \R \to Q$ is a curve in $Q$, the {\rm canonical
lifting} of $\phi$ to $\Tan^kQ$ is the curve $j^k\phi\colon
\Real\to\Tan^kQ$ defined as the $k$-jet lifting of $\phi$ restricted
to $\Tan^{k}Q \hookrightarrow J^{k}\pi$ (see
\cite{MdLRo}).

\section{The Hamilton-Jacobi problem in the Skinner-Rusk formalism}

Let $Q$ be a $n$-dimensional smooth manifold modeling the
configuration space of a $k$th-order autonomous dynamical system
with $n$ degrees of freedom, and let $\Lag \in \Cinfty(\Tan^kQ)$ be
a Lagrangian function for this system, which is assumed to be
regular. In the Lagrangian-Hamiltonian formalism, we consider the
bundle $\W = \Tan^{2k-1}Q \times_{\Tan^{k-1}Q} \Tan^*(\Tan^{k-1}Q)$
with canonical projections $\pr_1 \colon \W \to \Tan^{2k-1}Q$ and
$\pr_2 \colon \W \to \Tan^*(\Tan^{k-1}Q)$. It is clear from the
definition that the bundle $\W$ fibers over $\Tan^{k-1}Q$. Let $\p
\colon \W \to \Tan^{k-1}Q$ be the canonical projection. Obviously,
we have $\p = \rho^{2k-1}_{k-1} \circ \pr_1 = \pi_{\Tan^{k-1}Q}
\circ \pr_2$. Hence, we have the following commutative diagram
$$
\xymatrix{
\ & \W \ar[dl]_{\pr_1} \ar[dr]^{\pr_2} \ar[dd]_{\p} & \ \\
\Tan^{2k-1}Q \ar[dr]_{\rho^{2k-1}_{r}} & \ & \Tan^*(\Tan^{k-1}Q) \ar[dl]^-{\pi_{\Tan^{k-1}Q}} \\
\ & \Tan^{k-1}Q & \
}
$$

We consider in $\W$ the presymplectic form $\Omega =
\pr_2^*\,\omega_{k-1} \in \df^{2}(\W)$, where
$\omega_{k-1} \in \df^{2}(\Tan^*(\Tan^{k-1}Q))$ is the
canonical symplectic form. In addition, from the Lagrangian function
$\Lag$, and using the canonical coupling function $\C \in
\Cinfty(\W)$, we construct a Hamiltonian function $H \in
\Cinfty(\W)$ as $H = \C - \Lag$. Thus, the dynamical equation for
the system is
\begin{equation}\label{eqn:LagHamDynEq}
\inn(X_{LH})\Omega = \d H \ , \quad X_{LH} \in \vf(\W) \, .
\end{equation}
Following the constraint algorithm in [5], a solution to the
equation \eqref{eqn:LagHamDynEq} exists on the points of a
submanifold $j_o \colon \W_o \hookrightarrow \W$ which can be
identified with the graph of the Legendre-Ostrogradsky map $\Leg
\colon \Tan^{2k-1}Q \to \Tan^*(\Tan^{k-1}Q)$ associated to $\Lag$.
If the Lagrangian function is regular, then there exists a unique
vector field $X_{LH}$ solution to \eqref{eqn:LagHamDynEq} and
tangent to $\W_o$ (see \cite{SKRK}).

\subsection{The generalized Hamilton-Jacobi problem}

We first state the generalized version of the Hamilton-Jacobi
problem. Following the same patterns as in \cite{GHJT}, \cite{CoMdLPrRo} and
\cite{dLMdDVa} (see also an approach to the problem for
higher-order field theories in \cite{Vi-12}), the natural definition for the generalized
Hamilton-Jacobi problem in the Skinner-Rusk setting \cite{PrRR}, \cite{SKRK} is
the following.

\begin{definition}\label{def:GenLagHamHJDef}
The \textnormal{generalized $k$th-order Lagrangian-Hamiltonian
Hamilton-Jacobi problem} (or \textnormal{generalized $k$th-order unified
Hamilton-Jacobi problem}) consists in finding a section $s \in
\Gamma(\p)$ and a vector field $X \in \vf(\Tan^{k-1}Q)$ such that
the following conditions are satisfied:
\begin{enumerate}
\item The submanifold ${\rm Im}(s) \hookrightarrow \W$ is contained in $\W_o$.
\item If $\gamma \colon \R \to \Tan^{k-1}Q$ is an integral curve of $X$,
then $s \circ \gamma \colon \R \to \W$ is an integral curve of
$X_{LH}$, that is,
\begin{equation}\label{eqn:GenLagHamHJDef}
X \circ \gamma = \dot{\gamma} \Longrightarrow
X_{LH} \circ (s \circ \gamma) = \dot{\overline{s \circ \gamma}} \, .
\end{equation}
\end{enumerate}
\end{definition}

It is clear 
that the vector field $X \in \vf(\Tan^{k-1}Q)$ cannot be chosen independently from
the section $s \in \Gamma(\p)$. Indeed, following the same pattern as in \cite{GHJT} we can prove:

\begin{proposition}\label{prop:GenLagHamHJRelatedVF}
The pair $(s,X) \in \Gamma(\p) \times \vf(\Tan^{k-1}Q)$ satisfies
the two conditions in Definition \ref{def:GenLagHamHJDef} if, and
only if, $X_{LH}$ and $X$ are $s$-related.
\end{proposition}

\begin{corollary}\label{corol:GenLagHamHJRelatedVF}
If $s \in \Gamma(\p)$ and $X \in \vf(\Tan^{k-1}Q)$ satisfy the two
conditions in Definition \ref{def:GenLagHamHJDef}, then $X = \Tan\p \circ X_{LH} \circ s$.
\end{corollary}

That is, the vector field $X \in \vf(\Tan^{k-1}Q)$ is completely
determined by the section $s \in \Gamma(\p)$, and it is called the
\textsl{vector field associated to $s$}.
Therefore, the search of a pair $(s,X) \in \Gamma(\p) \times
\vf(\Tan^{k-1}Q)$ satisfying the two conditions in Definition
\ref{def:GenLagHamHJDef} is equivalent to the search of a section $s
\in \Gamma(\p)$ such that the pair $(s,\Tan\p \circ X_{LH} \circ s)$
satisfies the same condition. Thus, we can restate the problem as follows.

\begin{proposition}\label{def:GenLagHamHJSolution}
The generalized $k$th-order unified Hamilton-Jacobi problem
for $X_{LH}$ is equivalent to finding a section $s \in \Gamma(\p)$ satisfying
the following conditions:
\begin{enumerate}
\item The submanifold ${\rm Im}(s) \hookrightarrow \W$ is contained in $\W_o$.
\item If $\gamma \colon \R \to \Tan^{k-1}Q$
is an integral curve of $\Tan\p \circ X_{LH} \circ s \in
\vf(\Tan^{k-1}Q)$, then $s \circ \gamma \colon \R \to \W$ is an
integral curve of $X_{LH}$, that is
$$
\Tan\p \circ X_{LH} \circ s \circ \gamma = \dot{\gamma} \Longrightarrow
X_{LH} \circ (s \circ \gamma) = \dot{\overline{s \circ \gamma}} \, .
$$
\end{enumerate}
\end{proposition}

\begin{proposition}\label{prop:GenLagHamHJEquiv}
The following assertions on a section $s \in \Gamma(\p)$ are
equivalent.
\begin{enumerate}
\item $s$ is a solution to the generalized $k$th-order unified
Hamilton-Jacobi problem.
\item The submanifold ${\rm Im}(s) \hookrightarrow \W$ is invariant under
the flow of the vector field $X_{LH}$ solution to equation
\eqref{eqn:LagHamDynEq} (that is, $X_{LH}$ is tangent to the
submanifold ${\rm Im}(s)$).
\item The section $s$ satisfies the dynamical equation $\inn(X)(s^*\Omega) = \d(s^*H)$,
where $X = \Tan\p \circ X_{LH} \circ s$ is the vector field associated to $s$.
\end{enumerate}
\end{proposition}
\begin{proof}
The proof is analogous to that of Proposition 6 and Theorem 2 in \cite{GHJT}.
\end{proof}

\paragraph{Coordinate expression.}

Let $(q_0^A)$ be a set of local coordinates in $Q$, with $1
\leqslant A \leqslant n$, and
$(q_0^A,\ldots,q_{2k-1}^A,p_A^{0},\ldots,p_A^{k-1})$ the induced
local coordinates in $\W$ (see \cite{PrRR} for details). Then, local
coordinates in $\W$ adapted to the $\p$-bundle structure are
$(q_i^A,q_j^A,p_A^i)$, where $0 \leqslant i \leqslant k-1$, $k
\leqslant j \leqslant 2k-1$. Hence, a section $s \in \Gamma(\p)$ is
given locally by $s(q_i^A) = (q_i^A,s_j^A,\alpha_A^i)$, where
$s_j^A,\alpha_A^i$ are local functions in $\Tan^{k-1}Q$.

From Proposition \ref{prop:GenLagHamHJEquiv}, an equivalent condition for a section
$s \in \Gamma(\p)$ to be a solution of the generalized $k$th-order unified
Hamilton-Jacobi problem is that the dynamical vector field $X_{LH}$ is tangent
to the submanifold ${\rm Im}(s) \hookrightarrow \W$, which is defined
locally by the constraints $q_j^A - s_j^A = 0$  and $p_A^i - \alpha_A^i = 0$.
From \cite{PrRR}, the vector field $X_{LH}$ solution to equation \eqref{eqn:LagHamDynEq}
is given locally by
$$
X_{LH} = \sum_{l=0}^{2k-2}q_{l+1}^A\derpar{}{q_l^A} + F^A\derpar{}{q_{2k-1}^A}
+ \derpar{\Lag}{q_0^A}\derpar{}{p_A^0} + \left( \derpar{\Lag}{q_i^A} - p_A^{i-1} \right) \derpar{}{p_A^i} \, ,
$$
where $F^A$ are the functions solution to the following system of
$n$ equations
$$
(-1)^k(F^B - d_T(q_{2k-1}^B))\derpars{\Lag}{q_k^B}{q_k^A}
+ \sum_{i=0}^k(-1)^i d_T^i\left(\derpar{\Lag}{q_i^A}\right) = 0 \, .
$$
Hence, requiring $X_{LH}(q_j^A-s_j^A) = 0$ and $X_{LH}(p_A^i - \alpha_A^i) = 0$
we obtain the following system of $2kn$ partial differential equations on ${\rm Im}(s)$
\begin{equation}\label{eqn:GenLagHamHJLocal}
\begin{array}{l}
\displaystyle
s_{j+1}^A - q_{i+1}^B\derpar{s_j^A}{q_i^B} - s_{k}^B\derpar{s_j^A}{q_{k-1}^B} = 0 \ ; \
F^A - q_{i+1}^B\derpar{s_{2k-1}^A}{q_i^B} - s_{k}^B\derpar{s_{2k-1}^A}{q_{k-1}^B} = 0 \\[15pt]
\displaystyle
\derpar{\Lag}{q_A^0} - q_{i+1}^B\derpar{\alpha^0_A}{q_i^B} - s_{k}^B\derpar{\alpha_A^0}{q_{k-1}^B} = 0 \ ; \
\derpar{\Lag}{q_l^A} - \alpha_A^{l-1} - q_{i+1}^B\derpar{\alpha^l_A}{q_i^B} - s_{k}^B\derpar{\alpha^l_A}{q_{k-1}^B} = 0 \, .
\end{array}
\end{equation}
This is a system of $2kn$ partial differential equations with $2kn$
unknown function $s_j^A$, $\alpha_A^i$. Hence, a section $s \in
\Gamma(\p)$ is a solution to the generalized $k$th-order
Lagrangian-Hamiltonian Hamilton-Jacobi problem if, and only if, its
component functions satisfy the local equations
\eqref{eqn:GenLagHamHJLocal}.

\subsection{The Hamilton-Jacobi problem}
\label{sec:LagHamHJProblem}

In general, 
to solve the generalized $k$th-order Hamilton-Jacobi problem is a difficult task
since we must find $kn$-dimensional submanifolds of $\W$ contained
in the submanifold $\W_o$ and invariant by the dynamical vector field $X_{LH}$.
Hence, it is convenient to consider a less general problem and require some
additional conditions to the section $s \in \Gamma(\p)$   \cite{AM,GHJT}.

\begin{definition}\label{def:LagHamHJDef}
The \textnormal{$k$th-order Lagrangian-Hamiltonian Hamilton-Jacobi problem}
consists in finding sections $s \in \Gamma(\p)$ solution to the
generalized $k$th-order unified Hamilton-Jacobi problem
such that $s^*\Omega = 0$. Such a section is called a
\textnormal{solution to the $k$th-order Lagrangian-Hamiltonian Hamilton-Jacobi problem}.
\end{definition}

From the definition of $\Omega \in \df^{2}(\W)$ we have
$$
s^*\Omega = s^*(\pr_2^*\omega_{k-1}) = (\pr_2 \circ s)^*\omega_{k-1}\ .
$$
Hence, $s^*\Omega = 0$ if, and only if, $(\pr_2 \circ s)^*\omega_{k-1} = 0$.
As $\Gamma(\pi_{\Tan^{k-1}Q}) = \df^{1}(\Tan^{k-1}Q)$, the section
$\pr_2 \circ s \in \Gamma(\pi_{\Tan^{k-1}Q})$ is a $1$-form in $\Tan^{k-1}Q$,
and from the properties of the tautological form $\theta_{k-1}$
of the cotangent bundle $\Tan^*(\Tan^{k-1}Q)$ we have
$$
(\pr_2 \circ s)^*\omega_{k-1} = (\pr_2 \circ s)^*(-\d\theta_{k-1}) =
-\d((\pr_2 \circ s)^*\theta_{k-1}) = -\d(\pr_2 \circ s) \, .
$$
Hence, the condition $s^*\Omega=0$ is equivalent to
$\pr_2 \circ s \in \df^{1}(\Tan^{k-1}Q)$ being a closed $1$-form.
Therefore, the Hamilton-Jacobi problem can be reformulated as follows.

\begin{proposition}\label{def:LagHamHJDefClosedForm}
The $k$th-order unified Hamilton-Jacobi problem
is equivalent to finding sections $s \in \Gamma(\p)$ solution to the
generalized $k$th-order unified Hamilton-Jacobi problem
such that $\pr_2 \circ s$ is a closed $1$-form in $\Tan^{k-1}Q$.
\end{proposition}

Taking into account the new assumption $s^*\Omega = 0$ in Definition \ref{def:LagHamHJDef},
a consequence of Proposition \ref{prop:GenLagHamHJEquiv} is the following result.

\begin{proposition}\label{prop:LagHamHJEquiv}
The following assertions on a section $s \in \Gamma(\p)$
satisfying $s^*\Omega = 0$ are equivalent:
\begin{enumerate}
\item $s$ is a solution to the $k$th-order unified Hamilton-Jacobi problem.
\item $\d(s^*H) = 0$.
\item ${\rm Im}(s)$ is an isotropic submanifold of $\W$
invariant by $X_{LH}$.
\item The integral curves of $X_{LH}$ with initial conditions in ${\rm Im}(s)$
project onto the integral curves of $X = \Tan\p \circ X_{LH} \circ s$.
\end{enumerate}
\end{proposition}

\paragraph{Coordinate expression.}
From \cite{PrRR}, the Hamiltonian function
in $\W$ has coordinate expression
$H = q_{i+1}^Ap_A^i - \Lag(q_0^A,\ldots,q_k^A)$. Thus, its differential is given locally by
$$
\d H = -\derpar{\Lag}{q_0^A}\d q_0^A + \left( p_A^i - \derpar{\Lag}{q_{i+1}^A} \right)\d q_{i+1}^A
+ q_{i+1}^A\d p_A^i \, .
$$
Hence, the condition $\d(s^*H) = 0$ in Proposition \ref{prop:LagHamHJEquiv}
holds if, and only if, the following $kn$ partial differential equations
are satisfied
\begin{equation}\label{eqn:LagHamHJLocal}
\begin{array}{l}
\displaystyle q_{i+1}^B\derpar{\alpha^i_B}{q_0^A} + s_{k}^B\derpar{\alpha_B^{k-1}}{q_0^A} +
\alpha_B^{k-1}\derpar{s_{k}^B}{q_0^A} - \left( \derpar{\Lag}{q_0^A} + \derpar{\Lag}{q_{k}^B}\derpar{s_{k}^B}{q_0^A} \right) = 0 \, ,\\[15pt]
\displaystyle q_{i+1}^B\derpar{\alpha^i_B}{q_l^A} + s_{k}^B\derpar{\alpha_B^{k-1}}{q_l^A} + \alpha_A^{l-1} +
\alpha_B^{k-1}\derpar{s_{k}^B}{q_l^A} - \left( \derpar{\Lag}{q_l^A} + \derpar{\Lag}{q_{k}^B}\derpar{s_{k}^B}{q_l^A} \right) = 0 \, ,
\end{array}
\end{equation}
where $1 \leqslant l \leqslant k-1$.

Equivalently, we can require the $1$-form $\pr \circ s \in \df^{1}(\Tan^{k-1}Q)$
to be closed, that is, $\d(\pr \circ s) = 0$. Locally, this condition reads
\begin{equation}\label{eqn:LagHamHJLocalClosedForm}
\derpar{\alpha_A^i}{q_j^B} - \derpar{\alpha_B^j}{q_i^A} = 0 \ ,
\, \mbox{with } A \neq B \mbox{ or } i \neq j \, .
\end{equation}

Therefore, a section $s \in \Gamma(\p)$
is a solution to the $k$th-order Lagrangian-Hamiltonian Hamilton-Jacobi problem
if, and only if, the local functions $s_j^A,\alpha_A^i$ satisfy the system
of partial differential equations given by  \eqref{eqn:GenLagHamHJLocal}
and \eqref{eqn:LagHamHJLocal}, or, equivalently \eqref{eqn:GenLagHamHJLocal}
and \eqref{eqn:LagHamHJLocalClosedForm}. Observe that the system of partial differential equations
may not be $\Cinfty(U)$-linearly independent.

\subsection{Relation with the Lagrangian and Hamiltonian formalisms}

Finally, we state the relation between the solutions of the
Hamilton-Jacobi problem in the unified formalism and the solutions
of the problem in the Lagrangian and Hamiltonian settings given in
\cite{CoMdLPrRo}.

\begin{theorem}
Let $\Lag \in \Cinfty(\Tan^{k}Q)$ be a hyperregular Lagrangian
function.
\begin{enumerate}
\item If $s \in \Gamma(\p)$ is a solution to the (generalized)
$k$th-order Lagrangian-Hamiltonian Hamilton-Jacobi problem, then
the sections $s_\Lag = \pr_1 \circ s \in
\Gamma(\rho^{2k-1}_{k-1})$ and $\alpha = \pr_2 \circ s \in
\df^{1}(\Tan^{k-1}Q)$ are solutions to the (generalized)
$k$th-order Lagrangian and Hamiltonian Hamilton-Jacobi problems,
respectively.

\item If $s_\Lag \in \Gamma(\rho^{2k-1}_{k-1})$ is a solution to the
(generalized) $k$th-order Lagrangian
Hamilton-Jacobi problem, then 
$s = j_o \circ \overline{\pr}_1^{-1} \circ s_\Lag \in \Gamma(\p)$ is
a solution to the (generalized) $k$th-order
Lagrangian-Hamiltonian Hamilton-Jacobi problem.

\noindent If $\alpha \in \df^{1}(\Tan^{k-1}Q)$ is a solution to the
(generalized) $k$th-order Hamiltonian Hamilton-Jacobi problem, then
$s =j_o \circ \overline{\pr}_2^{-1} \circ \alpha \in \Gamma(\p)$ is
a solution to the (generalized) $k$th-order Lagrangian-Hamiltonian Hamilton-Jacobi problem.
\end{enumerate}
\end{theorem}

\begin{proof}
The proof of the first item follows the same patterns that the proof of Theorem
1 in \cite{CoMdLPrRo}.
For the second item, the key point is to take into account
that the maps $\overline{\pr}_1 \colon \W \to \Tan^{2k-1}Q$
and $\overline{\pr}_2 \colon \W \to \Tan^*(\Tan^{k-1}Q)$ are
diffeomorphisms, and that the dynamical vector field $X_{LH} \in \vf(\W)$
solution to equation \eqref{eqn:LagHamDynEq} is tangent to $\W_o$, and
therefore is $j_o$-related to a vector field $X_o \in \vf(\W_o)$
for which it is possible to state an equivalent Hamilton-Jacobi
problem.
\end{proof}

\subsection{An example: A (homogeneous) deformed elastic cylindrical beam with fixed ends}

Consider a deformed elastic cylindrical beam with both ends fixed
(see \cite{PrRR} and references therein). The problem is to determinate its shape; that is,
the width of every section transversal to the axis. This gives rise
to a $1$-dimensional second-order dynamical system, which is
autonomous if we require the beam to be homogeneous. Let $Q$ be the
$1$-dimensional smooth manifold modeling the configuration space of
the system with local coordinate $(q_0)$. Then, in the natural
coordinates of $\Tan^2Q$, the Lagrangian function for this system is
$$
\Lag(q_0,q_1,q_2) = \frac{1}{2}\mu q_2^2 + \rho q_0 \, ,
$$
where $\mu,\rho \in \R$ are constants, and $\mu\neq0$. This a
regular Lagrangian function because the Hessian matrix
$$
\ds \left( \derpars{\Lag}{q_2}{q_2} \right) = \mu \, ,
$$
has maximum rank equal to $1$ when $\mu \neq 0$.

In the induced natural coordinates $(q_0,q_1,q_2,q_3,p^0,p^1)$ of
$\W,$ the coordinate expressions of the presymplectic form $\Omega =
\pr_2^*\omega_1 \in \df^{2}(\W)$ and the Hamiltonian function $H =
\C - \Lag \in \Cinfty(\W)$ are
$$
\Omega = \d q_0 \wedge \d p^0 + \d q_1 \wedge \d p^1 \quad ; \quad
H = q_1p^0 + q_2p^1 - \frac{1}{2}\mu q_2^2 - \rho q_0 \, .
$$
Thus, the semispray of type $1$ $X_{LH} \in \vf(\W)$ solution to the
dynamical equation \eqref{eqn:LagHamDynEq} and tangent to the
submanifold $\W_o = {\rm graph}(\Leg) \hookrightarrow \W$ has the
following coordinate expression
$$
X_{LH} = q_1\derpar{}{q_0} + q_2\derpar{}{q_1} + q_3\derpar{}{q_2}
- \frac{\rho}{\mu}\derpar{}{q_3} + \rho\derpar{}{p^0} - p^0 \derpar{}{p^1} \, .
$$

In the following we state the equations for the (generalized)
Lagrangian-Hamilonian Hamilton-Jacobi problem for this dynamical
system.

In the generalized Lagrangian-Hamiltonian Hamilton-Jacobi problem we
look for sections $s \in \Gamma(\p)$, given locally by
$s(q_0,q_1) = (q_0,q_1,s_2,s_3,\alpha^0,\alpha^1)$, such that the
submanifold ${\rm Im}(s) \hookrightarrow \W$ is invariant under the
flow of $X_{LH} \in \vf(\W)$. Since the constraints defining locally
${\rm Im}(s)$ are $q_2-s_2=0$, $q_3-s_3=0$, $p^0-\alpha^0=0$,
$p^1-\alpha^1=0$, then the equations for the section $s$ are
$$
\begin{array}{l}
\displaystyle s_3 - q_1\derpar{s_2}{q_0} - s_2\derpar{s_2}{q_1} = 0 \ ; \
-\frac{\rho}{\mu} - q_1\derpar{s_3}{q_0} - s_2\derpar{s_3}{q_1} = 0 \, , \\[10pt]
\displaystyle \rho - q_1\derpar{\alpha^0}{q_0} - s_2\derpar{\alpha^0}{q_1} = 0 \ ; \
-\alpha^0 - q_1\derpar{\alpha^1}{q_0} - s_2\derpar{\alpha^1}{q_1} = 0 \, .
\end{array}
$$

For the Lagrangian-Hamiltonian Hamilton-Jacobi problem, we
require in addition the section $s \in \Gamma(\rho^\W_1)$ to satisfy
$s^*\Omega = 0$ or, equivalently, the form $\pr_2 \circ s \in
\df^{1}(\Tan Q)$ to be closed. In coordinates, if $s =
(q_0,q_1,s_2,s_3,\alpha^0,\alpha^1)$, then the $1$-form $\pr_2 \circ
s$ is given by $\pr_2\circ s = \alpha^0 \d q_0 + \alpha^1 \d q_1$.
Hence, a section $s \in \Gamma(\p)$ solution to the
unified Hamilton-Jacobi problem for this
system must satisfy the following system of $5$ partial differential
equations
$$
\begin{array}{c}
\displaystyle s_3 - q_1\derpar{s_2}{q_0} - s_2\derpar{s_2}{q_1} = 0 \ ; \
-\frac{\rho}{\mu} - q_1\derpar{s_3}{q_0} - s_2\derpar{s_3}{q_1} = 0 \ ; \
\displaystyle \derpar{\alpha^1}{q_0} - \derpar{\alpha^0}{q_1} = 0 \, , \\[10pt]
\displaystyle \rho - q_1\derpar{\alpha^0}{q_0} - s_2\derpar{\alpha^0}{q_1} = 0 \ ; \
-\alpha^0 - q_1\derpar{\alpha^1}{q_0} - s_2\derpar{\alpha^1}{q_1} = 0 \, .
\end{array}
$$
Observe that, when the condition ${\rm Im}(s) \subseteq \W_o = {\rm graph}\Leg$
is required, these equations project to the Lagrangian or Hamiltonian equations for
the Hamilton-Jacobi problem \cite{CoMdLPrRo}.

\section*{Acknowledgments}

We acknowledge the financial support of the \textsl{MICINN}
(Spain), projects MTM2010-21186-C02-01,
MTM2011-22585 and  MTM2011-15725-E; AGAUR, project 2009 SGR:1338.;
IRSES-project ``Geomech-246981''; and ICMAT Severo Ochoa project
SEV-2011-0087. P.D. Prieto-Mart\'{\i}nez wants to thank the UPC for
a Ph.D grant, and L. Colombo wants to thank CSIC for a JAE-Pre grant.

\end{document}